\documentclass[11pt]{article}

\usepackage{amssymb, amsmath, amsthm, amsfonts, pdfsync}
\usepackage{microtype, graphicx, subfigure}
\usepackage{hyperref}
\usepackage{macros}
\usepackage{fullpage}

\makeatletter
\renewcommand*{\@fnsymbol}[1]{\ensuremath{
  \ifcase#1\or * \or \natural \else\@ctrerr\fi}}
\makeatother

\newcommand{\gsim}{\stackrel{c}{>}}
\newcommand{\lsim}{\stackrel{c}{<}}
\newcommand{\ds}{\gamma_{s}}
\newcommand{\lips}{\mathcal{L}}
\newcommand{\Lw}{\lips_w}
\newcommand{\Lu}{\lips_u}

\title{The Well Tempered Lasso }
\author{
  Yuanzhi Li\thanks{
    Princeton University. Email: \textrm{yuanzhil@cs.princeton.edu}}
  \and
  Yoram Singer\thanks{
    Google Brain \& Princeton University.
    Email: \textrm{singer@google.come}}
}

\begin{document}

\maketitle

\begin{abstract}
We study the complexity of the entire regularization path for least squares
regression with $1$-norm penalty, known as the Lasso. Every regression
parameter in the Lasso changes linearly as a function of the regularization
value. The number of changes is regarded as the Lasso's complexity.
Experimental results using exact path following exhibit
polynomial complexity of the Lasso in the problem size. Alas, the path
complexity of the Lasso on artificially designed regression problems is
\emph{exponential}.

We use \emph{smoothed analysis} as a mechanism for bridging the gap
between worst case settings and the de facto low complexity. Our analysis
assumes that the observed data has a \emph{tiny} amount of intrinsic noise.
We then prove that the Lasso's complexity is polynomial in the problem size.
While building upon the seminal work of Spielman and Teng on smoothed
complexity, our analysis is morally different as it is \emph{divorced} from
specific path following algorithms. We verify the validity of our analysis in
experiments with both worst case settings and real datasets. The empirical
results we obtain closely match our analysis.
\end{abstract}

\section{Introduction}
In high dimensional learning problems, a sparse solution is often desired as
it has better generalization and interpretation. In order to promote sparse
solutions, a regularization term that penalizes for the $1$-norm of the vector
of parameters is often augmented to an empirical loss term. In regression
problems, the empirical loss amounts to sum of the squared differences between
the linear predictions and true targets. The task of linear regression with
$1$-norm penalty is known as Lasso~\cite{tibshirani1996regression}.  To obtain
meaningful solutions for the Lasso, it is required to pick a good value of the
$\ell_1$ regularizer.  To automatically choose the best regularization value,
algorithms that calculate all possible solutions were
developed~\cite{efron2004least, osborne2000new, tibshirani2012}.

These algorithms find the solution set for all possible regularization
values, commonly referred to as the entire regularization path. The
algorithms typically built upon the property that the Lasso regularization
path is piecewise linear in the constituents of solution vector. As a
result, their running times are also governed by the total number of linear
segments. While experiments with real datasets suggest that the number of
linear segments is in practice linear in the dimension of the
problem~\cite{rosset2007piecewise}, worst case
settings~\cite{mairal2012complexity} can yield to exponentially many linear
segments. The construction of exponentially complex regression problems
of~\cite{mairal2012complexity} stands in stark contrast to the
aforementioned methods. Provably polynomial path complexity has so far
derived in vastly more restricted settings, such as the one described
in~\cite{dubiner2013maxent}.

We bridge the gap between the de facto complexity of the regularization path
in real problems and the worst case analysis of the number of linear
segments.  We show that under fairly general models, the complexity of the
entire regularization path is guaranteed to be polynomial in the dimension
of the problem. As an important observation, settings which attain the worst
case complexity of the regularization path often exhibit fragile algebraic
structure. In contrast, natural datasets often comes with noise, which
renders those highly frail structures improbable.  This approach is called
the \emph{smoothed analysis}, introduced by the seminal paper
of~\cite{spielman2009smoothed}.

The core of smoothed analysis is the assumption that the input data is
subjected to a small \emph{intrinsic} noise. Such noise may come from
uncertainty in the physical measurements when collecting the data, irrational
decisions in human feedback, or simply the rounding errors in the computation
process used for obtaining the data.  In this model, we let
$\bX\in\mathbb{R}^{n\times{}d}$ denote the data matrix where $n$ is the number
of observations is $d$ is the dimension (number of free parameters).  Smoothed
analysis assumptions implies that $\bX$ is the sum of $\bhX$, an unknown
fixed matrix, and $\bG$, which consists of i.i.d. random samples from a normal
distribution with a zero mean and low variance, $\bX=\bhX+\bG$. In this view,
the data is neither \emph{completely random} nor \emph{completely arbitrary}.
The \emph{smoothed complexity} of the problem is measured as the expected
complexity taken over the random choices of $\bG$. Using this framework, it
was proved that the smoothed running time of the simplex, $k$-means, and the
Perceptron algorithm~\cite{spielman2001smoothed, arthur2009k,
blum2002smoothed} is in fact polynomial while the worst case complexity of
these problems is exponential.

We use the above smoothed analysis model to show that on ``typical''
instances, the total number of linear segments of the Lasso's \emph{exact}
regularization path is polynomial in the problem size with high probability.
Informally speaking and omitting technical details, our main result can be
stated as follows.
\begin{center}
\fbox{\parbox{0.86\textwidth}{
Let $\bX\in\mathbb{R}^{n \times d}$ be a data matrix of the form
$\bX = \bhX + \bG$ for \emph{any} fixed matrix $\bhX$ and a
random Gaussian matrix $\bG$, $\bG_{ij} \sim \mathcal{N}(0, \sigma^2)$.
Then, for an \emph{arbitrary} vector of targets $y \in \mathbb{R}^n$,
with high probability, the total number of linear segments of the
Lasso's \emph{exact} regularization path for $(\bX, y)$ is
\emph{polynomial} in $n$, $d$, and
$\frac1\sigma$.}}
\end{center}

Our result is conceptually different than the one presented
in~\cite{mairal2012complexity}. Mairal~and~Yu showed that
there exists an \emph{approximate} regularization path with a small number of
linear segments. However, the analysis, while being novel and inspiring, does
not shed light on why, in practice, the \emph{exact} number of linear segments
is small as the \emph{approximated} path is unlikely to coincide with the
\emph{exact} path. Our analysis covers uncharted terrain and different aspects
than the approximated path algorithms in~\cite{mairal2012complexity,
giesen2010approximating}. On one hand, we show that when the input data is
``typical'', namely comes from a ``naturally smooth'' distribution, then with
high probability, the total number of linear segments, of the \emph{exact}
Lasso path, is already polynomially small. This part of our analysis provides
theoretical backing to the empirical findings reported
in~\cite{rosset2007piecewise}. On the other hand, when the input matrix is
atypical and induces a high-complexity path, then we can also obtain a
low-complexity approximate regularization path by adding a small amount of
random noise to the data and then solve the Lasso's regularization path on the
perturbed instance \emph{exactly}.  We also verify our analysis experimentally
in section~\ref{sec:exp}. We show that even a \emph{tiny} amount of
perturbation to high-complexity data matrices, results in a dramatic drop in
the number of linear segments.

The technique used in this paper is morally different from the smoothed
analysis obtained for simplex~\cite{spielman2001smoothed},
k-means~\cite{arthur2009k}, and the Perceptron~\cite{blum2002smoothed}, as
there is no concrete algorithm involved. We develop a new framework which
shows that when the total number of linear segments is excessively high, then we
can \emph{tightly} couple the solutions of the original, smoothed problems to
another set of solutions in a manner that does not depend on $\bG$. We then
use the randomness of $\bG$ to show that such couplings are unlikely to exist,
thus high complexity solutions are rare. We believe that our framework can
be extended to other problems such as the regularization path of support
vector machines.

\section{Preliminaries} \label{prelim:sec}
We use uppercase boldface letters, e.g, $\bX$, to denote matrices and
lowercase letters $x, w$ to denote vectors, variables, and scalars. We use
$\bX_i$ to denote the $i$'th column of $\bX$. Given a set $\set{S}$, we
denote by $\bX_{\set{S}}\in\reals^{d\times|\set{S}|}$ the sub-matrix of
$\bX$ whose columns are $\bX_i$ for $i \in \set{S}$.  Analogously,
$\bX_{\set{\bar{S}}}$ denotes the sub-matrix of $\bX$ with columns $\bX_i$
for $i \notin \set{S}$. For a matrix $\bX \in \mathbb{R}^{n \times d}$ with
$n \geq d$, we use the term \emph{smallest} (\emph{largest}) singular value
of $\bX$ to denote the smallest (largest) right singular value of $\bX$. We
define the \emph{generalized sign} of a scalar $b$ as follows,
\begin{align*}
\sign(b) = \left\{
  \begin{array}{rl}
    +1& \mbox{$b > 0$} \\
    -1& \mbox{$b < 0$} \\
     0& \mbox{$b = 0$}
   \end{array} ~ .\right.
\end{align*}

Let $y$ be a vector in $\mathbb{R}^n$ and let $\bX = \left[\bX_1, \cdots,
\bX_d\right]$ be a matrix in $\mathbb{R}^{n \times d}$. The Lasso is the
following regression problem,
\begin{align}
w[\lambda] = \displaystyle \arg\!\!\min_{w \in \mathbb{R}^d } \,
  \frac{1}{2} \| \bX w - y \large\|_2^2 + \lambda \|w \|_1
  ~ . \label{Lasso:eqn}
\end{align}
Here, $\lambda > 0$ is the regularization value. The value of $\lambda$
influences the sparsity level of the solution $w[\lambda]$. The larger
$\lambda$ is the sparser the solution is. When $\bX$ is of full column rank,
the solution to \eqref{Lasso:eqn} is unique. We therefore denote it by
$w[\lambda]$. We use $$\mathcal{P} = \{ w[\lambda] \mid \lambda > 0\}$$ to
denote the set of all possible solution vectors. This set is also referred to
as the \emph{entire regularization path}.

To establish out main result we need a few technical lemmas.  The first
Lemma from~\cite{mairal2012complexity} provides optimality conditions for
$w[\lambda]$.
\begin{lem}\label{lem:opt}
Let $\lambda > 0$, the $w[\lambda]$ is the optimal solution iff
it satisfies the following conditions,
\begin{enumerate}
\item There exists a vector $u[\lambda]$ s.t.
  $$\bX^{\top}(\bX w[\lambda] - y) = u[\lambda] ~ .$$
\item Each coordinate of $u[\lambda]$ satisfies,
\begin{align*}
  u_i[\lambda] = \left\{
    \begin{array}{ll}
      -\lambda \sign\left(w_i[\lambda] \right) & |w_i[\lambda]| > 0\\
      \in [-\lambda, \lambda] & \mbox{o.w.}
    \end{array} \right. ~ .
\end{align*}
\end{enumerate}
\end{lem}

Let us denote the sign vector as $\sign(w[\lambda])$, which is obtained by
applying the generalized sign function $\sign(\cdot)$ element-wise to
$w[\lambda]$. The result of~\cite{mairal2012complexity} shows that
$\mathcal{P}$ is piecewise linear and unique in the following sense.
\begin{lem}\label{lem:piecewise_linear}
Suppose $\bX$ is of full column rank, then $\mathcal{P} = \{ w[\lambda] \mid
\lambda > 0\}$ is unique, well-defined, and $w[\lambda]$ is piecewise linear.
Moreover, for any $\lambda_1, \lambda_2 > 0$, if the sign vectors at
$\lambda_1$ and $\lambda_2$ are equal, $\sign(w[\lambda_1]) =
\sign(w[\lambda_2])$, then $w[\lambda_1]$ and $w[\lambda_2]$ are in the same
linear segment.
\end{lem}
We use $|\mathcal{P}|$ to denote the total number of linear segments in
$\mathcal{P}$.
We denote by $\alpha > 0$ the smallest singular value of $\bX$. Without
loss of generality, as we can rescale $\bX$ and $y$ accordingly, we assume
that $\| y \|_2 = 1$. To obtain our main result, we introduce the following
smoothness assumption on the data matrix $\bX$.
\begin{assu}[Smoothness]\label{ass:1}
$\bX$ is generated according to,
\begin{align*}\bX = \bhX  + \bG ~ ,\end{align*}
where $\bhX \in \mathbb{R}^{n \times d}$ ($n \geq d$) is a fixed unknown
matrix with $\| \bhX \|_2 \leq 1$. Each entry of $\bG$ is an i.i.d. sample
from the normal distribution with $0$ mean and variance of
$\frac{\sigma^2}{n}$.
\end{assu}
We use $\mathcal{N}(0, \sigma^2 / n)$ instead of $\mathcal{N}(0, \sigma^2)$
for the noise distribution $\bG$.
This choice implies that when $\sigma$ is a constant the spectral norm
of $\bG$ is also a constant in expectation, $\E[\| \bG \|_2] = O(1)$, see for
instance~\cite{rudelson2010non}. Therefore, the signal-to-noise ratio satisfies,
$$
\E\big[\|\bX\|_2 \|\bG\|_2^{-1}\big] = O(1) ~ .
$$

It is convenient to think of $\sigma$ as an arbitrary small constant. The
analysis presented in the sequel employs a fixed constant $c$ that does not
depend on the problem size. We use $f(\cdot) \gsim poly(\cdot)$
(analogously, $f(\cdot) \lsim poly(\cdot)$) to denote the fact that the
function $f$ is everywhere greater (smaller) than a polynomial function up
to a multiplicative constant. Equipped with the above conventions and the
smoothness assumption, the following lemma, due to~\cite{sankar2006smoothed},
characterizes the extremal singular values of $X$.
\begin{lem}\label{lem:alpha}
Let $\delta > 0$. With probability of at least $1-\delta$, the
smallest, denoted $\alpha$, and largest, denoted $\beta$, right
singular values of $\bX$ satisfy,
$$\alpha \gsim \frac{\delta\sigma}{d}
~~\mbox{and}~~
\beta \lsim 1 + \sigma\log({1}/{\delta}) ~ .$$
\end{lem}
The bound on $\beta$ lets assume henceforth that $\beta$ is $O(1)$ for any
reasonable choices of $\sigma$ and $\delta$. In our analysis We describe
explicitly dependencies on $\|X\|$ for clarification and states the main
results with $\beta=O(1)$.

\medskip

The main result of the paper is states in the following theorem.

\begin{thm}[Lasso's Smoothed Complexity]\label{thm:main}
Suppose assumption~\ref{ass:1} holds for arbitrary
$n, d \in \mathbb{Z}$ with $n \geq d$ and $\sigma\in ( 0, 1]$.
Then, with a probability of at least $1 - \delta$ (over the random selection
of $\bG$), the complexity of the Lasso satisfies,
\begin{align*}
  |\mathcal{P}| \lsim n^{1.1} \left(\frac{d}{\delta\sigma}\right)^6 ~ .
\end{align*}
\end{thm}

\section{Main Lemmas}

To prove the main theorem, we introduce several properties of
$\bX, y, w[\lambda]$ and $u[\lambda]$ that are critical in the
analysis of $|\mathcal{P}|$. We then use the smoothness
assumption to bound these properties.

\begin{defn}[Lipschitzness]
Let $w_{i}[\lambda]$ and $u_i[\lambda]$ be the value of the $i$'th
coordinate of $w[\lambda]$ and $u[\lambda]$ respectively for $i\in[d]$.
The coordinate-wise Lipschitz parameters of $w$ and $u$ are defined as,
\begin{align*}
\Lw = \max_{i \in [d]}
  \sup_{\lambda > 0} \Big| {\partial_{_\lambda} w_i[\lambda]} \Big|
~ ~ , ~ ~
\Lu =  \max_{i \in [d]}
  \sup_{\lambda > 0} \Big| {\partial_{_\lambda} u_i[\lambda]} \Big|
  ~ .
\end{align*}
\end{defn}
By definition, $\Lw$ and $\Lu$ characterize how much each coordinate of
$w[\lambda]$ and $u[\lambda]$ can change as we vary the value of $\lambda$.
We later use the smoothness assumption to show that $\Lw$ and $\Lu$ are
polynomially small. This implies that $w[\lambda]$ and $u[\lambda]$ would not
change too fast with $\lambda$. However, Lipschitzness by itself does not give
us a bound on $|\mathcal{P}|$ since $w[\lambda]$ can still oscillate around
\emph{zero} and induce an excessively large number of linear segments.
Therefore, we also need the following property which defines the restricted
distance between the column space of $\bX$ and $y$.
\begin{defn}[Subspace distance]\label{lem:sub_d}
For any $s, \delta > 0$,  let $\ds$ denote the largest value
such that,
 \begin{align*}
  \Pr\left[ \exists v \in \mathbb{R}^{d - s} \, \text{ s.t. }
    \left\|\bX_{\set{\bar{S}}} v - y \right\|_2 \leq
      \ds \right] \leq \delta ~ ,
  \end{align*}
for all $\set{S} \subset [d]$ of size $s$.
\end{defn}
This definition quantifies the distance of $y$ to a subspace spanned by
$s\leq{}d$ columns of $\bX$. Since $y \in \mathbb{R}^n$, $n \geq d$, and $\bX$
is smooth, it can be shown that $y$ cannot be too close to the subspace
spanned by $\bX_{\set{\bar{S}}}$. That is, $ \ds$ is inversely proportional to
a polynomial in $n, d, 1/\sigma$.  We interchangeably use in the following the
original matrix $\bX$ with $v\in\reals^d$ s.t. $v_i=0$ for $i\in\set{S}$ and
$\bX_{\set{\bar{S}}}$ with $v\in\reals^{d-s}$. Using the above properties, we
prove the following theorem.
\begin{thm}[Exact Smooth Complexity]\label{thm:tec_main}
Let $\bX$ satisfy Assumption~\ref{ass:1}.
Then, for all $s\in[d]$ and
$\delta>0$, with probability of at least $1 - \delta$ the complexity of the
Lasso satisfies,
\begin{align*}
|\mathcal{P}| \lsim
  3^s \left(
      \frac{\sqrt{sn} d \left(\frac{ \Lw}{\alpha^2} + \Lu \right)}
           {\delta^2 \sigma \ds}
    \right)^\frac{s}{s - 1}
  ~ .
\end{align*}
\end{thm}

The following lemma characterizes the (smoothed) values of
$\Lw$, $\Lu$, and $\ds$.
\begin{lem} \label{lem:LD}
Let $\bX$ satisfy Assumption~\ref{ass:1}. Then with
probability of at least $1 - \delta$ the following properties hold,
\begin{align*}
  \Lw, \Lu  \lsim \frac{\sqrt{d}}{\alpha^2} ~~ , ~~
  \ds       \gsim \frac{\sigma}{\sqrt{dn } (d/\delta)^{2/s}} ~~ .
\end{align*}
\end{lem}

Applying the bounds on $\Lw\,$, $\Lu\,$, $\ds\,$, and $\alpha$ to
Theorem~\ref{thm:tec_main} while letting $s$ be a sufficiently large constant,
we can directly prove Theorem~\ref{thm:main}. In Section~\ref{sec:1}, we use
the value of $\alpha$ to bound $\Lw$ and $\Lu$. In Section~\ref{sec:2}, we
employ the smoothness of $\bX$ to bound $\ds$. Finally, in Section~\ref{sec:3}
we prove Theorem~\ref{thm:tec_main}.

\section{Proof sketch}
Since $\|\bX\|_2 = O(1)$, there exists a constant $\lambda_{\max} = \Omega(1)$
such that for $\lambda\geq\lambda_{\max}$, $w[\lambda]$ is the zero vector.
Thus, we can divide $\lambda\in[0,\lambda_{\max}]$ into
${\lambda_{\max}}/{\dlam}$ intervals, each of size $\dlam$ for some (inversely
polynomial) small $\dlam$. We then show that within every interval the total
number of linear segments exceeds a fixed polynomial number with exponentially
small probability. The total number of linear segments follows by taking a
union bound over all intervals. We need to specifically address the following
two questions in the analysis.

\subsection*{What if $|\mathcal{P}|$ within an intervals is excessively large?}
We will show that when there are $N$ linear segments in an interval, then there
must be at least $\log_3(N)$ many coordinates of $w[\lambda]$ , which we
denote as the set $\set{S}$, that change their sign. Since $\dlam$
is small and $w[\lambda]$ is a Lipschitz function in $\lambda$, we know that
those coordinates of $w[\lambda]$ must be close to zero. Therefore, we can
show that $w[\lambda]$ is close to the optimal solution, $v[\lambda]$, of the
Lasso problem when the entries of coordinates in $\set{S}$ are constrained to
be \emph{exactly} zero,
\begin{align}
  v[\lambda] = \displaystyle \arg\!\!\min_{w\in\reals^d} \,
  \frac{1}{2} \| \bX w - y \large\|_2^2 + \lambda \|w \|_1 ~ \mbox{ s.t. } ~
  \forall i\in\set{S}:\, w_i=0
  ~ . \label{modlasso:eqn}
\end{align}

\subsection*{What if $w[\lambda]$ oscillates excessively around $v[\lambda]$?}
From the optimality condition of $u[\lambda]$ and the smoothness of
$u[\lambda]$, we also know that the coordinates in $\set{S}$ of $u[\lambda]$
must be close to either $-\lambda$ or $\lambda$. Thus,
$u_{\set{S}}[\lambda] \eqdef
  \bX_{\set{S}}^{\top} (\bX\,w[\lambda] - y)$ is close
to a vector on the scaled hypercube $\{ -\lambda, \lambda\}^{| \set{S} |}$.
On the other hand,
if $w[\lambda]$ is close to $v[\lambda]$, we know that
$\bX_{\set{S}}^{\top} (\bX\,v[\lambda] - y)$ is also close to a vector in
$\{ -\lambda, \lambda\}^{| \set{S} |}$.
However, $v[\lambda]$ does not depend on $\bX_{\set{S}}$ by construction.
Therefore, the residual $\bX v[\lambda] - y$ also does not
depend on $\bX_{\set{S}}$. Thus, using the randomness etched in $\bX_{\set{S}}$
we can now show that $\bX_{\set{S}}^{\top} (\bX v[\lambda] - y)$ is close to
a vector in $\{ -\lambda, \lambda\}^{| \set{S} |}$ with probability which is
exponentially small in the size of $\set{S}$.  Therefore, we know that w.h.p.
the total number of linear segments in this interval is unlikely to be large.

\section{Bounding $\Lw$ and $\Lu$}\label{sec:1}
Recall that we denote the smallest singular value of $\bX$ by $\alpha$. We
first show the following lemma regarding perturbations of strongly convex
functions. Also recall that a second-order smooth function
$f:\mathbb{R}^d\to\mathbb{R}$ is \emph{$\alpha^2$-strongly convex} if
$\nabla^2 f(x) \geq \alpha^2$ for all $x \in \mathbb{R}^d$ and is
$\lips$-Lipschitz if $\|\nabla{}f(x)\|_{2}\leq{}\lips$ for all $x \in \mathbb{R}^d$.

\begin{lem}[Perturbation of strongly convex functions I]\label{lem:pscf_I}
Let $f: \mathbb{R}^d \to \mathbb{R}$ be an non-negative, $\alpha^2$-strongly
convex function. Let $g: \mathbb{R}^d \to \mathbb{R}$ be a $\lips$-Lipschitz
non-negative convex function . For any $\lambda \geq 0$, let $z[\lambda]$
be the minimizer of $f(z) + \lambda g(z)$, then we have,
$$\left\| \frac{d z[\lambda]}{d\lambda} \right\|_2 \leq \frac{\lips}{\alpha^2} ~.$$
\end{lem}

\begin{proof}[\bf{Proof of Lemma~\ref{lem:pscf_I}}]
For any $\tau \geq 0$ and $\lambda \geq 0$, let us abbreviate
$z = z[\lambda]$ and denote $\veps = z[\lambda + \tau] - z$.
From $\alpha^2$-strong convexity of $f$ at $z $ and the
optimality of $z$ at $\lambda$, we know that
\begin{align}\label{eq:fahfoasu}
\frac{1}{2}\alpha^2 \| \veps \|_2^2 + f(z) +
  \lambda g(z) \leq f(z + \veps) + \lambda g(z + \veps) ~.
\end{align}
Moreover, using the optimality of $z + \veps$ at $\lambda + \tau$,
we know that,
\begin{align}
&\frac{1}{2}\alpha^2 \| \veps \|_2^2 +
  f(z + \veps) + (\lambda + \tau) g(z + \veps) 
\leq f(z ) + (\lambda + \tau) g(z )\label{eq:fahfoasu2} ~ .
\end{align}
Summing Eqs.~\eqref{eq:fahfoasu} and \eqref{eq:fahfoasu2} and rearranging
terms yields,
\begin{align*}
\alpha^2 \| \veps \|_2^2  \leq \tau \left( g(z) - g(z + \veps) \right)
  \leq \tau \| \veps \|_2 \lips ~,
\end{align*}
where the last inequality is due to the Lipschitzness assumption on $g$.
Therefore, we get that ${\|\veps\|_2}/{\tau} \leq {\lips}/{\alpha^2}$. Letting
$\tau \to 0^+$ completes the proof.
\end{proof}
Using Lemma~\ref{lem:pscf_I} with $f(w)  = \frac{1}{2} \| \bX w - y \|_2^2$
and $g(w) = \| w\|_1$, we obtain Lipschitz properties for $w[\lambda]$ and
$u[\lambda]$. Since we assume that the minimum singular value of $\bX$ is
$\alpha$ then $f(w)$ is $\alpha^2$-strongly convex. In addition, the norm
of $\nabla g(w)$ is clearly at most $\sqrt{d}$.
To simplify notation, when $w[\lambda]$ is not differentiable
at a point $\lambda$, we define ${d w[\lambda]}/{d \lambda} = 0$.
Due to Lipschitzness and strong convexity all vectors in the subgradient
set $\partial_\lambda\,w[\lambda]$ include this particular choice for
a subgradient. In summary we get the following corollary.
\begin{cor}[Lipschitzness of $w$] \label{lem:lip_w}
For any $\lambda \geq 0$ it holds that,
$$ \left\|\frac{d w[\lambda]}{d \lambda} \right\|_2  \leq
  \frac{\sqrt{d}}{~\alpha^2} ~ .  $$
\end{cor}
\noindent
Since by definition, $u[\lambda] = \bX_i^{\top} (\bX w[\lambda] - y )$,
we obtain a similar corollary for $u$.
\begin{cor}[Lipschitzness of $u$]
For any $\lambda \geq 0$ it holds that,
$$\left\|\frac{d u[\lambda]}{d \lambda}
  \right\|_2 \leq \frac{\| \bX \|_2^2 \sqrt{d}}{\alpha^2} ~ .$$
\end{cor}
\noindent Recall that $\|\bX\|_2=O(1)$, thus, from the above corollaries
we get
\begin{equation} \label{eqn:lips_bounds}
\Lw,\Lu\lsim\sqrt{d}/\alpha^2 ~ .
\end{equation}
We also use in the next section
the following Lemma.
\begin{lem}[Perturbation of strongly convex functions II] \label{lem:pscf_II}
Let $f: \mathbb{R}^d \to \mathbb{R}$ be an $\alpha^2$-strongly convex
function and $g: \mathbb{R}^d \to \mathbb{R}$ an $\lips$-Lipschitz convex
function. Let $z_1$ and $z_2$ be the minimizers of $f(z)$ and $f(z) + g(z)$,
respectively, then
$$\left\| z_1 - z_2 \right\|_2  \leq \frac{\lips}{\alpha^2} ~ .$$
\end{lem}
\begin{proof}[\bf{Proof of Lemma \ref{lem:pscf_II}}]
Let $\veps = z_2 - z_1$. From strong convexity of $f$ and optimality of
$z_1$, we know that
\begin{align}\label{eq:fahfoasu3}
  \frac{1}{2}\alpha^2 \| \veps \|_2^2 + f(z_1) \leq f(z_2)  ~.
\end{align}
Due to the optimality of $z_2$ for $f(z) + g(z)$ we get,
\begin{align}\label{eq:fahfoasu4}
\frac{1}{2}\alpha^2 \| \veps \|_2^2 + f(z_2) +g(z_2) \leq f(z_1) + g(z_1)
~ .
\end{align}
Summing Eq.~\eqref{eq:fahfoasu3} and \eqref{eq:fahfoasu4} and rearranging
terms gives,
\begin{align*}
\alpha^2 \| \veps \|_2^2  \leq  g(z_1) - g(z_2) \leq \lips \| \veps \|_2  ~.
\end{align*}
We therefore get that $\| \veps \|_2 \leq \frac{\lips}{\alpha^2}$.
\end{proof}

\section{Bounding $\ds$}\label{sec:2}
In this section we prove that Assumption~\ref{ass:1} also yields a bound
on $\ds$.

\begin{lem}\label{lem:sub_distance}
Let $y \in \mathbb{S}^{n-1}$ be an arbitrary unit vector. Then for any
$s \in \{10,\dots,d\}$ and a set $\set{S} \subset [d]$ of size $s$, the
following holds,
$$
\Pr\left[
    \exists v \in \mathbb{R}^{d - s} \; \mbox{ s.t. } \;
    \left\|\bX_{\set{\bar{S}}} v - y \right\|_2
      \lsim \frac{\sigma {\delta}^{\frac1s}}{\sqrt{dn}}
  \right] \leq \delta ~ .
$$
\end{lem}
\noindent For brevity, we denote the event above by $\set{D}_{\gamma}$ where
$\gamma = \Omega\left(\frac{\sigma\delta^{\frac1s}}{\sqrt{dn}} \right)$.
\begin{proof}
The proof relies on the following simple fact about Gaussian random
variables. For any unit vector $u$, scalar $\tau > 0$, $g \in \mathbb{R}$, and $i \in [d]$
the following holds,
\begin{align}
\Pr\big[\left|\langle u, \bX_i \rangle - g \right| \leq \tau \big]
    \leq \frac{e \tau \sqrt{n}}{\sigma} ~ ,
  \label{lem:div}
\end{align}
where $\bX_i$ is a (column) vector with elements distributed i.i.d according
to ${\cal N}(0,\sigma^2/n)$.

Let us assume that there exists $v \in \mathbb{R}^{d - s}$ such that
$\left\|\bX_{\set{\bar{S}}} v - y \right\|_2 \leq \gamma$.  From
Lemma~\ref{lem:alpha}, we know that with probability of at least
$1 - \delta$, $ \|\bX\|_2  \lsim 1 $. Since we assume that $\|y\|_2=1$,
it holds that
\begin{align*}
\gamma &
  \geq \left\|\bX_{\set{\bar{S}}} v - y \right\|_2
  \geq \|y\|_2 - \|\bX v\|_2
  \geq 1 - \|v\|_2 ~ ~ \Rightarrow ~ ~
  \| v\|_2 \geq 1-\gamma ~ .
\end{align*}
Since $\gamma$ is smaller than $1/2$,  there must exist a coordinate $i$ for which
$| v_i| \geq \frac{1}{2 \sqrt{d}}$.

Now, let us fix $\bG_j$ to its observed values for all $j \not= i$. In
addition, let us denote by $\bU$ the subspace spanned by
$$\{(\bhX)_i\} \cup \{\bX_j \mid \forall j \not= i,j \in\set{S}\} \cup \{y\} ~.$$
We know that $\bU $ is of dimension $d - s + 1$. Consider an orthonormal basis
for $\{u_1, \cdots, u_{s - 1} \} \in \mathbb{R}^n$ for the subspace  $\mathrm{span}(\{\bX_i\}) - \bU$. Suppose there
  exists $v$ such that $ \left\|\bX_{\set{\bar{S}}} v - y \right\|_2 \leq \gamma$.
Then, for every $j \in [s-1]$, multiplying by $u_j$ yields
\begin{align*}
\big| \langle u_j, \bX_i \rangle v_i \big| \leq \gamma ~ \Rightarrow ~
| \langle u_j, \bX_i \rangle  | \leq \frac{\gamma}{|v_i|} \leq 2 \sqrt{d} \gamma ~ .
\end{align*}
Note that any pair $j\neq\ j'$, the inner products $ \langle u_j, \bX_i \rangle $  and
$ \langle u_{j'}, \bX_i \rangle $ are independent. We now use \eqref{lem:div}
over all $u_j$ and get that  $\set{D}_{\gamma}$ holds with probability of at most,
$$
\left( \frac{6\sqrt{dn}\gamma}{\sigma} \right)^{\large s} ~ .
$$
Finally, choosing
$\gamma = \Omega\Big(\frac{\sigma \delta^{\frac 1 s}}{\sqrt{dn}}\Big)$
completes the proof.
\end{proof}

\section{Coupling the solutions}\label{sec:3}
In this section, we show that when the total number of linear segments in
a small interval is excessively large, the optimal solution $w[\lambda]$
can be coupled with the optimal solution $v[\lambda]$ of the constrained
Lasso problem of \eqref{modlasso:eqn}.

\paragraph{Sign changes.}
For a given fixed $\lambda_0>0$ and $\dlam > 0$, let us denote by $\nsc(i)$
the number of times that the generalized sign of $w_i$ changes as $\lambda$
increases from $\lambda_0$ to  $\lambda_0 + \dlam$. Thus, the total number
of linear segments in the interval $[\lambda_0, \lambda_0 + \dlam]$ is at
least $\sum_{i = 1}^d \nsc(i) $. We prove the following lemmas related to
the sign changes.

\begin{lem}[Number of sign changes]
For any integer $N > 0$, any $\lambda_0, \dlam > 0$, if
$\sum_{i=1}^d\nsc(i)\geq{}N$, then there exists at least $\log_3 (N)$ many
indices $j \in [d]$ such that $\nsc(j) \geq 1$.
\end{lem}
\begin{proof}
According to Lemma~\ref{lem:piecewise_linear}, each linear
segment is associated with a {\em unique} sign pattern in $\{-1, 0, 1\}^d$.
Since there are $N$ segments, the pigeon hole principle implies that there
must exist at least $\log_3(N)$ many coordinates of $w[\lambda]$ that change
their sign in this interval.
\end{proof}
\noindent
From the Lipschitzness of $w[\lambda]$ and $u[\lambda]$, we also obtain the
following lemma.
\begin{lem}[Sign change $\Rightarrow$ small weight]
For $i \in [d]$, if $\nsc(i) \geq 1$, then following properties hold:
$ \left| w_i[\lambda_0]\right| \leq \Lw \dlam$ and
$ \left| |u_i[{\lambda_0}]| - \lambda_0 \right| \leq (\Lu+1) \dlam$.
\end{lem}
\begin{proof}
Since $\nsc(i) \geq 1$, we know that there exists  $\lambda \in
[\lambda_0, \lambda_0 + \dlam]$ such that $w_i[\lambda] = 0$ and
$\left|u_i[{\lambda}]\right|  = \lambda$. Using Lipschitzness of $w[\lambda]$
we get
$|w_i[\lambda_0]-w_i[\lambda]| = |w_i[\lambda_0]| \leq \Lw\dlam$. For
$u_i[\lambda_0]$ we have,
\begin{align*}
||u_i[\lambda_0] |- \lambda_0| &\leq
  ||u_i[\lambda_0]| - |u_i[\lambda] ||  +  | |u_i[\lambda] | - \lambda | + | \lambda - \lambda_0|
 \leq ( \Lu + 1) \dlam
\end{align*}
which concludes the proof.
\end{proof}

We use $\tilde{\Lu}$ in the sequel as a shorthand for $\Lu+1$.  Based on the
two lemmas above we readily get the following corollary.
\begin{cor} \label{cor:event}
For any integer $N > 0$ and $\dlam, \lambda_0 \geq 0$, if
$\sum_{i} \nsc(i) \geq N$, then there exists a subset $\set{S} \subseteq [d]$ of
cardinality at least $\log_3 (N)$ such that $\forall i \in \set{S}$,
\begin{align*}
\big|w_i[\lambda_0]\big| \leq \Lw \dlam & ~~ \mbox{and} ~~
\big| |u_i[{\lambda_0}]| - \lambda_0 \big| \leq \tilde{\Lu} \dlam ~.
\end{align*}

\end{cor}


\paragraph{Rare events.}
Let $\set{S}$ be defined as in Corollary~\ref{cor:event}, we next show that if
the size of $\set{S}$ is too large, then certain rare couplings would take
place. Thus, the size of $\set{S}$ is likely to be small with high
probability.  Throughout the rest of the paper we {\em overload} notation and
denote by $\bX_{\set{S}}\in\reals^{n\times{}d}$ the matrix where each column
$\bX_i$, for $i\not\in\set{S}$, is replaced with the zero
vector. The matrix $\bX_{\set{\bar{S}}}$ is defined analogously. Note that
by definition, $\bX_{\set{S}} + \bX_{\set{\bar{S}}}= \bX $.

\begin{lem}[Size of $\set{S}$]\label{lem:event}
For any fixed $\lambda_0$, $\dlam > 0$, and a set $\set{S} \subseteq [d]$,
let $\bX_{\set{\bar{S}}} \in \mathbb{R}^{n \times d}$ be defined as above.
Let $v_{ \set{\bar{S}}}[\lambda_0]$ be the minimizer,
$$v_{ \set{\bar{S}}}[\lambda_0] = \arg\min_{v\in\reals^d}
  \frac12 \| \bX_{\set{\bar{S}}}\, v- y \|_2^2 + \lambda_0 \| v \|_1$$
Assume that the properties of Corollary~\ref{cor:event}
hold for a set $\set{S}$. Then, for every $j \in \set{S}$ the following
inequality holds,
\begin{align*}
\left|\left|
  \bX_j^{\top} (\bX_{\set{\bar{S}}}\, v_{ \set{\bar{S}}}[\lambda_0] - y )
    \right|- \lambda_0 \right|
&\leq
  2\sqrt{|\set{S}|}\,\|\bX\|_2^2
  \left(\frac{\Lw \| \bX\|_2^2 \dlam}{\alpha^2} + \tilde{\Lu} \dlam \right) ~ .
\end{align*}
We refer to this event as $\set{E}_{\set{S}}^{(\tau)}$ with parameter
$\tau = 2\sqrt{|\set{S}|} \|\bX\|_2^2
\left(\frac{\Lw \| \bX \|_2^2 \dlam}{\alpha^2} + \tilde{\Lu} \dlam \right)$.
\end{lem}
\begin{proof}
We know that the $i$'th coordinate of $v_{\set{\bar{S}}}[\lambda_0]$ is
zero for all $i \in \set{S}$. Therefore, we need to focus solely on
the set of vectors $v$ which are in
$$\set{K} \eqdef
  \{ v \in \mathbb{R}^d \mid \forall i \in \set{S}, v_i = 0 \} ~ .$$
Since by definition $\bX_{\set{S}} = \bX - \bX_{\set{\bar{S}}}$
we can rewrite the original objective as,
$$\frac12 \big\| \bX_{\set{\bar{S}}} \, w + \bX_{\set{S}} \, w\,-\, y \big\|_2^2 +
  \lambda_0 \| w \|_1 ~ .$$
Let $w_{\set{\bar{S}}}[\lambda_0] \in \mathbb{R}^d$ be a vector whose
$i^{\textrm{th}}$ coordinate is the $i^{\textrm{th}}$ coordinate of
$w[\lambda_0]$ for $i \notin S$ and is zero otherwise and let
$w_{\set{S}}[\lambda_0]  = w - w_{\set{\bar{S}}}[\lambda_0] $. From the
optimality of $w[\lambda_0]$, we know that $w_{\set{\bar{S}}}[\lambda_0]$
is the minimizer of
$$h(w) =
  \frac12 \big\|\bX_{\set{\bar{S}}}\, w + \bX_{\set{S}} \,w_{\set{S}}[\lambda_0] \,-\, y\big\|_2^2
    + \lambda_0 \| w \|_1 \quad \mbox{s.t. } w \in \set{K} ~.$$
Let $g(w) \eqdef \frac12 \big\|\bX_{\set{\bar{S}}}\, w\,-\,y\big\|_2^2 + \lambda_0\|w\|_1$.
Expanding terms we get that for every $w \in \set{K}$,
\begin{align*}
h(w) - g(w)  &= \frac12 \left\| \bX_{\set{S}}\, w_{\set{S}}[\lambda_0]  \right\|_2^2
  - \langle  \bX_{\set{S}}\, w_{\set{S}}[\lambda_0] , y \rangle
  + \left\langle \bX_{\set{S}} w_{\set{S}}[\lambda_0] ,
    \bX_{\set{\bar{S}}} \, w  \right\rangle ~ ,
\end{align*}
and the gradient of $h(w) - g(w)$ satisfies,
\begin{align*}
\left\|\nabla\big((h(w) - g(w)\big) \right\|_2 &\leq
  \left\|\bX_{\set{\bar{S}}}^{\top}\, \bX_{\set{S}}\, w_{\set{S}}[\lambda_0]\right\|_2
   \; \leq \|\bX\|_2^2 \, \big\|w_{\set{S}}[\lambda_0]\big\|_2 ~ .
\end{align*}
From our assumption that Corollary~\ref{cor:event} holds for $\set{S}$, we know
that $\|w_{\set{S}}[\lambda_0]  \|_{\infty} \leq  \Lw \dlam$ which in turn
implies that $ \|w_{\set{S}}[\lambda_0]  \|_2 \leq \sqrt{|\set{S}|} \Lw \dlam$.

\smallskip

We can now apply Lemma~\ref{lem:pscf_II} w.r.t $g(w)$ and $h(w) - g(w)$ to conclude that
$$\left\| v_{\set{\bar{S}}} [\lambda_0] - w_{\set{\bar{S}}}[\lambda_0] \right\|_2
  \leq \|\bX\|_2^2  \frac{\sqrt{|\set{S}|} \Lw \dlam}{\alpha^2}~.$$
Therefore, for every $j \in \set{S}$ the following holds
\begin{align*}
\lefteqn{\Big| \left| \bX_j^{\top}
  (\bX_{\set{\bar{S}}}\, v_{\set{\bar{S}}}[\lambda_0] \,-\, y) \right|- \lambda_0 \Big|}
\\
  & \leq \Big| \left| \bX_j^{\top} (\bX_{\set{\bar{S}}}\, w_{ \set{\bar{S}}}[\lambda_0] - y )  \right|
  -  \lambda_0 \Big|  +  \| \bX \|_2^4   \frac{\sqrt{|\set{S}|} \Lw \dlam}{\alpha^2}
\\
  & \leq \Big| \left| \bX_j^{\top} (\bX_{\set{\bar{S}}} \, w[\lambda_0] - y )  \right|-  \lambda_0 \Big|  \nonumber
~+   \| \bX \|_2^2  \left(\frac{\sqrt{|\set{S}|}\, \| \bX \|_2^2 \, \Lw \,\dlam}{\alpha^2}  +
  \sqrt{|\set{S}|}\, \left\|w_{\set{S}}[\lambda_0] \right\|_{\infty}  \right)
\\
& \leq 2\sqrt{|\set{S}|} \, \| \bX \|_2^2 \, \left(\frac{\Lw\, \| \bX \|_2^2 \,\dlam}{\alpha^2} +
\tilde{\Lu} \dlam \right) ~,
\end{align*}
which concludes the proof.
\end{proof}

\smallskip

\noindent
Using again that $\|\bX\|_2 \lsim 1$ we obtain
\begin{align} \label{eqn:taunu}
\tau &\lsim
  \nu \sqrt{|\set{S}|} \left(\frac{\Lw}{\alpha^2} + \tilde{\Lu} \right) ~.
\end{align}
This means that of gradient of objective scales as the product
of the root of the size of $\set{S}$ and the length of the interval $\dlam$.

\paragraph{Bounding the probability of bad events}
Next we show that for every fixed set $\set{S}$ of sufficiently large
cardinality and sufficiently small $\dlam$, $\set{E}_{\set{S}}^{(\tau)}$
holds with very small probability if $\bX$ satisfies the smoothness
assumption.

\medskip

\begin{lem}[Smoothing] \label{lem:smoothed}
For any fixed $\lambda_0$, $\dlam > 0$, $\tau \geq 0$, and
$\set{S} \subseteq [d]$, let us decompose $\bG$ into
  $\bG = \bG_{\set{\bar{S}}} + \bG_{\set{S}}$ as in Lemma~\ref{lem:event}.
Then, the following inequality holds,
$$\Pr_{\bG_{\set{S}}}
  \left[\set{E}_{\set{S}}^{(\tau)} \big| \bG_{\set{\bar{S}}}\right] \leq
  \left(\frac
  {e\tau \sqrt{n}}
  {\sigma \big\| \bX_{\set{\bar{S}}}\, v_{\set{\bar{S}}}[\lambda_0] - y  \big\|_2}
    \right)^{|\set{S}|} ~ . $$
\end{lem}
\begin{proof}
Consider the vector $v_{\set{\bar{S}}}[\lambda_0]$ defined as in
Lemma~\ref{lem:event}. We know that $v_{\set{\bar{S}}}[\lambda_0] $
depends only on $ \bG_{\set{\bar{S}}} $ but not on $\bG_{\set{S}}$. Therefore,
for any fixed $ \bG_0$, by conditioning on $ \bG_{\set{\bar{S}}} = \bG_0$, we
get that for all $j \in \set{S}$
\begin{align*}
  &\Pr_{\bG_{\set{S}}} \Big[ \left|\,\left| \bX_j^{\top}
  (\bX_{\set{\bar{S}}} v_{ \set{\bar{S}}}[\lambda_0] - y )  \right|-
  \lambda_0 \right| \leq \tau \Big]
\\
& = \Pr_{\bG_{\set{S}}} \left[ \left| \left| \left((\bhX)_j + \bG_j
  \right)^{\top} (\bX_{\set{\bar{S}}} v_{ \set{\bar{S}}}[\lambda_0] - y )  \right|-  \lambda_0 \right| \leq \tau\right]
\\
& = \Pr_{\bG_j} \left[ \left| \left| \left((\bhX)_j + \bG_j \right)^{\top}
  (\bX_{\set{\bar{S}}} v_{ \set{\bar{S}}}[\lambda_0] - y )  \right|-  \lambda_0 \right| \leq \tau\right]
\\
  & \leq \frac{e\tau \sqrt{n}}
    { \sigma \big\| (\bX_{\set{\bar{S}}} v_{ \set{\bar{S}}}[\lambda_0] - y ) \big\|_2 } ~ .
\end{align*}
Since the inequality holds for all $j \in \set{S}$ and any $\bG_0$ the proof
is completed.
\end{proof}

\section{Proof of Theorem~\ref{thm:tec_main}}
Recall that we assume that the target vector is of unit norm $\|y\|_2 = 1$.
We slightly overload notation and denote by $\alpha(\bX)$ the smallest right
singular value of $\bX$. From the optimality of $w[\lambda]$, we know that
\begin{align*}
\frac{1}{2} \big\|\bX w[\lambda] - y\big\|_2^2 +
  \lambda \big\|w[\lambda]\big\|_1
  \leq \frac{1}{2} \|  y \|_2^2  \leq \frac{1}{2} ~ ,
\end{align*}
which implies that $\|\bX w[\lambda] - y \|_2 \leq 1$. From necessary
conditions for optimality we also get,
\begin{align*}
  \bX^{\top} (\bX w[\lambda] - y ) =  u[\lambda] ~ .
\end{align*}
Thus, we have
$$\| u[\lambda] \|_2  = \left\| \bX^{\top} (\bX w[\lambda] - y ) \right\|_2
  \leq \|\bX\|_2 = O(1) ~ .
$$
Therefore, there exists a constant
$\lambda_{\textrm{max}}$ such that implies that $w[\lambda] = 0$ for
$\lambda \geq \lambda_{\textrm{max}}$.

We next employ the randomness of $\bG$. Consider a fixed $\alpha_0$ and
examine the event that there exists one set $\set{S}$ of
size at least $s$ such that $\set{E}_{\set{S}}^{(\tau)}$  is true,
then it holds that,
\begin{align*}
& \Pr\left[\left(\exists \set{S}: |\set{S}| = s, \set{E}_{\set{S}}^{(\tau)}
\text{ holds}  \right) \cap \set{D}_{\gamma} \cap \alpha(\bX) \geq \alpha_0 \right]
\\
 &\leq \sum_{\set{S}_0 \subseteq [d], |\set{S}_0| =  s}
 \Pr\left[\big(\set{E}_{\set{S}}^{(\tau)} \text{ holds}, \set{S} =
 \set{S}_0\big) \cap \set{D}_{\gamma} \cap \alpha(\bX) \geq \alpha_0\right]
\\
&\leq  \sum_{\set{S}_0 \subseteq [d], |\set{S}_0| =  s}
\Pr\left[\set{E}_{\set{S}_0}^{(\tau)} \cap \set{D}_{\gamma} \cap \alpha(\bX) \geq \alpha_0 \right]
\\
& \leq  { d \choose s}   \left( \frac{ e \tau \sqrt{n}}{\sigma \gamma }\right)^{s}  \leq \left( \frac{ e \tau d \sqrt{n}}{\sigma \gamma }\right)^{s} ~ ,
\end{align*}
where we used the definition of $\set{D}_{\gamma}$ and the fact that
$\alpha\eqdef\alpha(\bX)\geq\alpha_0$ to obtain the last inequality.
We now set
$$\tau = O\left(\left(\delta \dlam \right)^{1/s}
  \frac{\sigma \gamma}{ d \sqrt{n}} \right) ~,$$
which in turn implies that (see \eqref{eqn:taunu}),
$$\dlam^{1 - 1/s} \gsim
  \frac{\delta^{1/s} \sigma\gamma}{
    d\sqrt{n s} \left({\Lw}/{\alpha_0^2} + \tilde{\Lu}\right)}
    ~~ \Rightarrow ~~
  \dlam \gsim \left(
  \frac{\delta^{1/s} \sigma\gamma}{
    d\sqrt{n s} \left({\Lw}/{\alpha_0^2} + \tilde{\Lu}\right)}
    \right)^{\frac{s}{s-1}}
    ~.$$
and obtain
\begin{align*}&
\Pr\left[\left(\exists \set{S}: |\set{S}| = s, \set{E}_{\set{S}}^{(\tau)}
\text{ holds}  \right) \cap \set{D}_{\gamma} \cap \alpha(\bX) \geq \alpha_0 \right]
  \leq  {\delta \dlam} ~ .
\end{align*}
Since we have at most $1/{\dlam}$ many intervals, taking union bound
over all linear segments we get that
$$ \Pr\left[ \left(N(\mathcal{P})\geq  \frac{3^s}{\dlam}\right) \cap
\set{D}_{\gamma} \cap\left( \alpha(\bX) \geq \alpha_0\right) \right]  \leq \delta $$
Finally, for properly chosen $\gamma$ and $\alpha_0$ we also obtain
$\Pr[\neg\set{D}_{\gamma}  \cup \left(\alpha(\bX) < \alpha_0 \right)] \leq 2\delta$
which completes the proof. \qed

\medskip

To recap, there exists a {\em universal} constant $c$ such that for {\em all}
$s\in[d]$ the complexity of the Lasso path is bounded above by,
\begin{align*}
  c\,3^s \left(
      \frac{\sqrt{sn} d \left(\frac{ \Lw}{\alpha^2} + \Lu \right)}
           {\delta^2 \sigma \ds}
    \right)^\frac{s}{s - 1} ~.
\end{align*}
We now use the bounds on $\Lw$, $\Lu$, and $\alpha$, yielding,
\begin{align*}
  \frac{\Lw}{\alpha^2} + \Lu \lsim \frac {d^{4.5}}{\delta^4 \sigma^4} ~ ~ , ~ ~
  \ds \lsim \frac{\sigma\delta^{\frac1s}}{\sqrt{dn}} ~ ~ ~ \Rightarrow ~ ~ ~
|\mathcal{P}| \lsim 3^s
  \left(
    \frac{\sqrt{s}\,n\,d^6}{\delta^{6+\frac1s}\,\,\sigma^6}
  \right)^\frac{s}{s-1} ~.
\end{align*}
By choosing
$s = O\left(\log\left(\frac{n d}{\delta \sigma}\right)\right)$ we get that
\begin{align*}
  |\mathcal{P}| \lsim n^{1.1} \left(\frac{d}{\delta\sigma}\right)^6 ~ .
\end{align*}

\begin{figure}[!ht]
 \begin{minipage}{0.475\linewidth}
    \centerline{\includegraphics[width=1.15\textwidth]{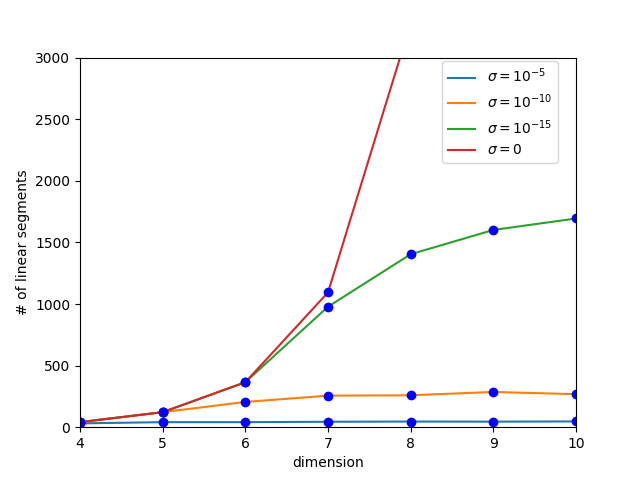}}
    \caption{\small Path complexity as a function of dimension for different levels
      of smoothing. The theoretical non-smoothed ($\sigma=0$) complexity is
      exponential in the size of the problem.}
 \end{minipage}
  \hfill
 \begin{minipage}{0.475\linewidth}
    \centerline{\includegraphics[width=1.15\textwidth]{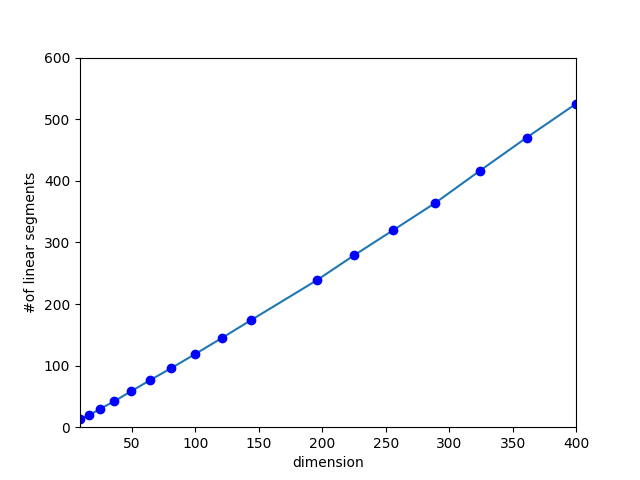}}
    \vspace{0.30cm}
    \caption{\small Path complexity in a regression task that predicts the value of a
      pixel from its neighboring pixels using the $\bold{MNIST}$ dataset.}
 \end{minipage}
\end{figure}

\section{Experiments}\label{sec:exp}
We performed two sets of experiments. Our path following implementation uses
Python with $\textbf{Float128}$ for high accuracy computations. In the first
set of experiments, we start with the exponential complexity construction
for $\bhX \in \mathbb{R}^{d \times d}$ from~\cite{mairal2012complexity},
which has $(3^d + 1)/2$ many line segments. We artificially added to each
entry of $\bhX$ i.i.d. Gaussian noise of mean zero and variance $\sigma^2$.
In this setting, the largest value of the entries of $\bhX$ is $1$ and $y$
is an all-one vector. We show the effect of dimension $d$ and smoothing
$\sigma$ on ${N}(\mathcal{P})$. We report the average over $100$ random
choices for smoothing per $\bhX$. As can be seen from the figure below, even
for a tiny amount of entry-wise noisy of $10^{-10}$, the number of linear
segments dramatically shrinks. We also include a full table of results,
where $1/$SNR denotes $-\log_{10}(\sigma)$.

\begin{table}[ht]
\begin{center}
{\large
\resizebox{0.48\textwidth}{!}{%
    \begin{tabular}{|r|r|r|r|r|r|r|r|}
    \hline {$1/$SNR} &
      $d = 4$ & $d = 5$ & $d = 6$ & $d = 7$ & $d = 8$ & $d = 9$ & $d   = 10$ \\
    \hline ${0}$ & 6& 8 & 10 & 12 & 13 & 15 & 18 \\
    \hline ${1}$ &  8& 9 & 10 &13& 16 & 14 & 17 \\
    \hline ${2}$ & 10& 13 & 13 & 15& 17 & 21 & 24\\
    \hline ${3}$ & 16& 18 & 18 & 20 & 25 &24 & 26\\
    \hline ${4}$ & 21& 26 & 27 & 27 & 30 &33 & 33 \\
    \hline ${5}$ & 31& 41&41 &44 & 46 &45  & 47\\
    \hline ${6}$ & 36& 50 & 55 & 58 & 59 &66 & 67\\
    \hline ${7}$ & 41 & 71 & 81 & 91 & 89 & 91 & 98\\
    \hline ${8}$ &  41 & 91 & 118 & 136 & 133 &134 & 138 \\
    \hline ${9}$ &  41 & 110 & 148 & 181 & 184 & 189 & 192\\
    \hline ${10}$ &  41 & 122 & 205 & 256 & 259 & 286 & 268 \\
    \hline ${11}$ &  41 & 122 & 276 &364 & 407 & 410 & 411 \\
    \hline ${12}$ &  41 & 122 & 354 & 467 & 538 & 560  & 566\\
    \hline ${13}$ &  41 & 122 &365 & 642 & 704 & 795 & 831\\
    \hline ${14}$ &  41 & 122 &365 & 872 & 1088 &1141 & 1165\\
    \hline ${15}$ &  41 & 122 &365 &978 & 1404 &1601& 1694 \\
    \hline ${16}$ &  41 & 122 &365 &1094 & 1814 &2162 & 2416\\
    \hline ${17}$ &  41 & 122 &365 &1094 & 2478 &3046 & 3343 \\
    \hline ${18}$ &  41 & 122 &365 &1094 &3030 &3894 & 4345\\
    \hline ${19}$ &  41 & 122 &365 &1094 &3281 &5323 & 6048\\
    \hline ${20}$ &  41 & 122 &365 &1094 &3281 & 7137 & 8592\\
    \hline $\infty$ &  41 & 122 &365 &1094 &3281 & 9842  & 29525 \\
    \hline
    \end{tabular}
}
}
\end{center}
\caption{Path complexity in a worst-case synthetic setting. The
theoretical non-smoothed complexity is exponential in the size of
the problem.}
\end{table}

\medskip

For the next experiment we use the $\bold{MNIST}$ data set. We randomly
selected $n = 1000$ images from the data set. We constructed the data matrix
$\bX \in \mathbb{R}^{n \times d^2}$ such that the $i$'th row of $\bX$ is a
randomly chosen patch from the $i$'th image of size $d \times d$.  We cast the
center pixel of patch $i$ as the target $y_i$ and discard the pixel from
$\bX$. Thus, the regression task amounts to predicting the center pixel $y_i$
using its surrounding pixels $\bX^{(i)}$. We plot the relation between the
size of the patch and the path complexity $N(\mathcal{P})$. Each point in the
graph is the average over $100$ random samples of patches and images. As
can be seen, when the amount of noise in the data is fixed and governed by the
data acquisition process (the minimum pixel value is $0$ and maximum is $255$
for MNIST), the number of linear segments increases barely faster than
\emph{linearly} in the dimension.

\section{Conclusions}
We proved that the smoothed complexity of the Lasso's regularization path is
pragmatically polynomial in the input size. Our analysis contrasts worst case
settings for which the Lasso's complexity is known to be exponential. To
illustrate the key idea, we provided analysis when smoothing each entry in the
data matrix by adding small amount of Gaussian noise. Although not presented
here, our analysis carries over to settings in which the smoothing is
performed using other distributions which can be sub-Gaussian,
sub-exponential, or even non i.i.d. so long as the rows of $\bG$ are
statistically independent.

\medskip

The nature of smoothed analysis usually imposes a
large polynomial factor~\cite{spielman2001smoothed}, as we do not make any
additional assumptions on the hidden matrix $\bhX$. However, we believe that
the polynomial degree in our results can be further reduced. For example, in
our proof, we used the fact that the condition number of $\bX$ is $\sim
\frac{1}{d}$, which is close to being ill-conditioned. For well-conditioned
matrices the polynomial bound can be improved to $O(d^{2.1} n^{1.1})$. This
reduction is also valid in settings when $n \geq 2d$ (see
e.g.~\cite{rudelson2010non} for different behaviors of the condition number of
Gaussian random matrices for $n = d$ and $n \geq 2d$).  Furthermore, when
$n\geq{}2d$, we can also improve $\ds$ to $\Omega(1)$, hence our
polynomial dependency can be further reduced to $O(d^{1.6} n^{0.6})$.  A final
improvement may stem from $\Lw$ and $\Lu$. We actually proved that
$$\left\|\frac{d w[\lambda]}{d \lambda } \right\|_2 =
  O\left( \frac{\sqrt{d}}{\alpha^2} \right) ~, $$
while we only need to use the infinity norm
$\left\|\frac{d w[\lambda]}{d \lambda } \right\|_{\infty}$. We leave these
improvements and further generalizations to future research.

\section*{Acknowledgements}
We would like to thank Vineet Gupta for thoughtful comments and feedback.

\bibliography{paper_arxiv}
\bibliographystyle{plain}

\end{document}